\documentclass{article}

\usepackage{fullpage} 

\usepackage{graphicx}
\usepackage{amsthm}

\usepackage{subfigure} 
\usepackage[ruled,vlined,linesnumbered]{algorithm2e} 

\usepackage{amsmath}
\usepackage{amssymb}
\usepackage{amstext}

\usepackage[font=small]{caption}
\usepackage{hyperref} 
\usepackage{microtype} 
\usepackage{times} 

\newtheorem{observation}{Observation}
\newtheorem{theorem}{Theorem}
\newtheorem{lemma}{Lemma}

\title{Embedding Four-directional Paths on Convex Point Sets\thanks{%
A shorter version of this paper will appear at Graph Drawing 2014~\cite{ahlmv-efdpc-14}.
O.A. supported by the ESF EUROCORES programme EuroGIGA - ComPoSe,
Austrian Science Fund (FWF): I 648-N18.
T.H.\ supported by the Austrian Science Fund (FWF): P23629-N18 `Combinatorial Problems on Geometric Graphs'.
}}
\author{Oswin Aichholzer\thanks{Institute for Software Technology, Graz University of Technology, Austria, {\tt [oaich|thackl|apilz|bvogt]@ist.tugraz.at}.}
 \and Thomas Hackl\footnotemark[2]
 \and Sarah Lutteropp\thanks{Institute of Theoretical Informatics, Karlsruhe Institute of
Technology, Germany \tt{sarah.lutteropp@student.kit.edu,
  mched@iti.uka.de}.}
 \and Tamara Mchedlidze\footnotemark[3]
 \and Birgit Vogtenhuber\footnotemark[2]}

\begin{document}

\maketitle

\begin{abstract}
A directed path whose edges are assigned labels ``up'', ``down'', ``right'', or ``left'' is called \emph{four-directional}, and \emph{three-directional} if at most three out of the four labels are used.
A \emph{direction-consistent embedding} of an \mbox{$n$-vertex} four-directional path $P$ on a set $S$ of $n$ points in the plane
is a straight-line drawing of $P$ where each vertex of $P$ is mapped to a distinct point of $S$ and every edge points to the direction specified by its label.
We study planar direction-consistent embeddings of three- and four-directional paths and provide a complete picture of the problem for convex point sets.
\end{abstract}

\section{Introduction}
\label{sec:introduction}
In 1974, Rosenfeld proved that every tournament has a spanning \emph{antidirected} path~\cite{Rosenfeld74} and conjectured that there exists an integer $n_0$ such that every tournament with more than $n_0$ vertices contains every \emph{oriented path} as a spanning subgraph. A tournament is a digraph whose underlying undirected structure is a complete graph and an oriented path is a digraph whose underlying undirected structure is a simple path.
An oriented path is antidirected if  the directions of its edges alternate. During the following decade several simplifications of  Rosenfeld's conjecture had been shown to be true. Alspach and Rosenfeld~\cite{AlpachR81} and Straight~\cite{Straight80} settled the conjecture for oriented paths with either a single source or a single sink. Forcade~\cite{Forcade73} proved the conjecture to be true for every tournament whose size is a power of two. Reid and Wormald\cite{ReidW83} showed that any tournament of size $n$ contains every oriented path of size $2n/3$ and Zhang~\cite{Zhang85} improved this result to $n-1$. Finally, in 1986, the conjecture was established by Thomason~\cite{thomason86}.

More than two decades later, with the expansion of Geometric Graph Theory and Graph Drawing, a geometric counterpart of Rosenfeld's conjecture was considered. The subject of this study is an \emph{upward geometric} tournament, that is, a tournament drawn on the plane with straight-line edges so that each edge points in the upward direction. It was asked whether an upward geometric tournament contains a planar copy of any oriented path~\cite{BinucciGDEFKL10}. Despite several independent approaches to attack the problem by different research groups, this question is still unsolved.   However, it 
was answered in the affirmative for several special cases of paths and tournaments. We use the following definitions to list these results. A vertex of a digraph which is either a source or a sink is called a \emph{switch}.  An oriented path whose edges are all oriented in the same direction is called \emph{monotone}. For the following cases it was shown that every upward tournament contains a planar copy of each oriented path: the vertices of the tournament are in convex position~\cite{BinucciGDEFKL10},
the oriented path has at most $3$ switches~\cite{BinucciGDEFKL10}, the oriented path has at most $5$ switches and at least two of its
monotone subpaths contain a single edge~\cite{AngeliniFGKMS10}, the oriented path where every sink is directly followed by a source~\cite{BinucciGDEFKL10}. It was also shown that each oriented path of size $n$ is contained in any upward geometric tournament of size $n 2^{k-2}$, where $k$ is the number of switches~\cite{AngeliniFGKMS10}. This result was later improved to $(n\!-\!1)^2\!+\!1$ in~\cite{Mchedlidze14}.
Recently, with the help of a computer, we could verify that every upward geometric tournament of size $10$ contains a planar copy of any oriented path as a spanning subgraph. This was done by exhaustive testing of all distinct directed order types, that is, all order types~\cite{Aichholzer01} with an additional combinatorial upward direction.

The question whether any upward geometric tournament contains a planar copy of any oriented path was originally stated in terms of so-called \emph{point set embeddings}. Here we are given a set $S$ of $n$ points in the plane and a planar $n$-vertex graph $G$, and we are asked to determine whether $G$ has a planar straight-line drawing where each vertex of~$G$ is mapped to a distinct point of $S$. This problem has been extensively studied and many exciting facts were established, see for example~\cite{BannisterCDE14,BiedlV12,Cabello06,DurocherM12,GritzmannMPP91}. In the upward counterpart of point set embeddings, $G$ is an upward planar digraph and the obtained drawing is additionally required to be upwards oriented. Such a drawing, if it exists, is called an \emph{upward point set embedding}.
Upward point set embeddings have been studied for different classes of digraphs~\cite{AngeliniFGKMS10,BannisterDE13,BinucciGDEFKL10,KaufmannMS13}.
Observe that the question whether any upward geometric tournament contains a planar copy of any oriented path is equivalent to asking whether any oriented path has an upward planar embedding on any set of $n$ points. We will refer to the latter as the \emph{oriented path question}.

The number of distinct plane embeddings of an (undirected) spanning path on a point set could provide us some additional evidence for the oriented path question. It is not difficult to see that if $S$ is a set of $n$ points in convex position, then it admits $n2^{n-3}$ distinct plane spanning path embeddings. Further it is known that this is the minimum number of distinct plane spanning path embeddings that a point set can admit, i.e., convex point sets minimize this number~\cite{ahhhkv-npg-06a}. Comparing this lower bound with the number of distinct oriented paths, which is $2^{n-1}$, it sounds even surprising that every oriented path has an upward planar embedding on every convex point set~\cite{BinucciGDEFKL10}. In order to approach the oriented path question in its general form, we aim to understand better how the nature of the problem changes when in addition to planarity of a path one requires its upwardness. To this end, we generalize the oriented path problem with respect to the number of considered directions (see Section~\ref{sec:definitions} for a rigorous definition). Observe that, instead of considering an oriented path, one can consider a monotone path with labels on edges that declare whether an edge is required to point up or down. In this work we study monotone paths with four possible labels on the edges: up, down, left, and right. We call such paths \emph{four-directional}, and \emph{three-directional} if at most three out of the four labels are used. An embedding of such a path on a point set where each edge points into the direction specified by its label is called \emph{direction-consistent}.
We study planar direction-consistent embeddings of three- and four-directional paths on convex point sets. Recall that convex point sets are extremal in the sense that they minimize the number of plane embeddings of (undirected) spanning paths. We provide a complete picture regarding four-directional paths and convex point sets. Our results are as follows:
\begin{itemize}
\item Every three-directional path admits a planar direction-consistent embedding on any convex point set.
\item There exists a four-directional path $P$ and a one-sided\footnote{A convex point set is called \emph{one-sided} if all of its points lie on the same side of the line through its bottommost and topmost points.}	convex point set $S$ such that $P$ does not admit a  planar direction-consistent embedding on $S$.
On the other hand, a four-directional path always admits a planar direction-consistent embedding for special cases of one-sided point sets, namely so-called quarter-convex point sets.

\item Given a four-directional path $P$ and a convex point set $S$, it can be decided in $O(n^2)$ time whether $P$ admits a planar direction-consistent embedding on $S$.
\end{itemize}
Our study is also motivated by applications similar to those of upward point set embeddings, i.e., any situation where a hierarchical structure must be represented and additional constraints on the positions of vertices are given. Our scenario, where instead of two directions the edges can point into four directions, allows for a more detailed control over a drawing.

The remainder of the paper is organized as follows. In Section~\ref{sec:definitions}, we give the necessary definitions. In Section~\ref{sec:preliminaries}, we prove several preliminary results which are utilized in our main Section~\ref{sec:three-dir}, where the existence of a planar direction-consistent embedding of a three-directional path on a convex point set is shown. All results on four-directional paths are concentrated in Section~\ref{sec:four-dir}. 
%

\section{Definitions}
\label{sec:definitions}

\paragraph{\bf Graphs}
The graphs we study in this paper are directed and we denote by $(u,v)$ an edge directed from $u$ to $v$. A directed edge when drawn as a straight-line segment is said to \emph{point up} or being \emph{upward}, if its source is below its sink. Similarly we define the notions of pointing \emph{down}, \emph{left}, and \emph{right}.
Our study concentrates on directed paths each edge of which is assigned one of four labels $U,D,L,R$, which means that (when the path is embedded on a point set) this edge is required to point up, down, left, or right, respectively. For simplicity, we will denote such a path containing vertices $v_1,\dots,v_n$ by $P=d_1,\dots,d_{n-1}$, where $d_i \in \{U,D,L,R\}$, $1\leq i \leq n-1$.
Let $T \subseteq \{U,D,L,R\}$. If $d_i \in T$, $1\leq i \leq n-1$, then $P$ is called \emph{$T$-path} and \emph{$|T|$-directional path} in order to emphasize the number of directions it contains. We denote by $P_{i,j}=d_i,\dots,d_j$, $1\leq i\leq j \leq n-1$, a subpath of~$P$. In addition, we define $P_{i,i-1}=v_i$.

\paragraph{\bf Point sets}
We say that a set $S$ of points in the plane is in \emph{general position} if no three points are collinear and no two points have the same $x$- or $y$-coordinate. All point sets mentioned in this paper are in general position.
Let $S$ be a convex point set. We denote by $\ell(S)$, $r(S)$, $t(S)$, $b(S)$ the leftmost, the rightmost, the topmost, and the bottommost point of $S$, respectively.
A subset of points of $S$ is called (\emph{clockwise}) \emph{consecutive} if its points appear consecutively as we (clockwise) traverse the convex hull of $S$.

A convex point set $S$ is called \emph{left-sided (resp. right-sided)} if $t(S)$ and $b(S)$ (resp. $b(S)$, $t(S)$) are clockwise consecutive on $S$. Further, $S$ is called \emph{one-sided} if $S$ is left-sided or right-sided. Finally, $S$ is called \emph{strip-convex} if $(i)$ the points $b(S)$ and $\ell(S)$ are either consecutive or coincide, and $(ii)$ the points $t(S)$ and $r(S)$ are either consecutive or coincide.
For $p,q \in S$, the points of $S$ which lie between the vertical lines through $p$ and $q$ (including them) are said to be \emph{vertically between} $p$ and $q$.

\paragraph{\bf Embeddings} Let $P$ be an $n$-vertex path (labeled) with vertex set $V(P)$ and $S$ be a set of $n$ points in general position. An \emph{embedding} of $P$ on $S$ is an injective function $\mathcal{E} \colon V(P) \to S$. If the edges of $P$ are drawn as  straight-line segments connecting corresponding end-vertices, the embedding $\cal E$ yields a drawing of $P$.  We say that the embedding $\cal E$ is \emph{planar} if this drawing is planar. We say that $\cal E$ is \emph{direction-consistent} if each edge points to the direction corresponding to its label. Planar direction-consistent embeddings are abbreviated by PDCE. During the construction of an embedding, a point $p$ is called \emph{used} if a vertex has already been mapped to it. Otherwise, $p$ is called \emph{free}. Throughout the paper we consider embeddings of $n$-vertex paths on sets of $n$ points, unless explicitly stated differently.

\paragraph{\bf Operations with paths, point sets, and embeddings}

Let $T \subseteq \{U,D,R,L\}$ and consider a $T$-path $P=d_1d_2\dots d_{n-1}$.
Let $S$ be a set of $n$ points and let $\cal E$ be a direction-consistent embedding of $P$ on $S$. Observe that $\cal E$ describes a direction-consistent embedding of another path $P^{\cal I}$ on the same point set $S$. Path $P^{\cal I}$ is called the \emph{reverse} path of $P$, and is constructed by reversing the directions of the edges of $P$ and changing the labels to their opposite. Thus, formally $P^{\cal I}={\cal I}(d_{n-1})\dots{\cal I}(d_2){\cal I}(d_{1})$, where ${\cal I}(U)=D$, ${\cal I}(D)=U$, ${\cal I}(R)=L$, and ${\cal I}(L)=R$. This embedding of $P^{\cal I}$ on $S$ is denoted by ${\cal E}^{\cal I}$. For example, if $P=UUDRL$, then $P^{\cal I}=RLUDD$. Observe also that $({P}^{\cal I})^{\cal I}=P$.
\begin{observation}
\label{obs:reverse}
Let $\cal E$ be a PDCE of a path $P$ on a point set $S$. Then ${\cal E}^{\cal I}$ is a PDCE of $P^{\cal I}$ on the same point set $S$.
\end{observation}

Let $P$, $S$, and $\cal E$ be as above. The embedding $\cal E$ yields a straight-line drawing $\Gamma$ of $P$.  Consider the rotation of $\Gamma$ counterclockwise by $\pi/2$. This rotated drawing represents a direction-consistent embedding, denoted by ${\cal R}({\cal E})$, of a new path, denoted by ${\cal R}(P)$, on the rotated point set, denoted by ${\cal R}(S)$. This new path ${\cal R}(P)$ is formally defined as follows:  ${\cal R}(P)={\cal R}(d_1){\cal R}(d_2)\dots {\cal R}(d_{n-1})$, where ${\cal R}(U)=L$, ${\cal R}(D)=R$, \mbox{${\cal R}(R)=U$}, and ${\cal R}(L)=D$.  We use the notation ${\cal R}^k$ for $k$ applications of $\cal R$. Thus, ${\cal R}^4(P)=P$ and ${\cal R}^4(S)=S$. Also, if $P$ is an $\{U,D,L\}$-path and $S$ is a right-sided point set then ${\cal R}^2(P)$ is an $\{U,D,R\}$-path and ${\cal R}^2(S)$ is a left-sided point set.
Note that $P^{\cal I}\neq {\cal R}^2(P)$.
\begin{observation}
\label{obs:rotation}
Let $\cal E$ be a PDCE of a path $P$ on a point set $S$. Then ${\cal R}({\cal E})$ is a PDCE of ${\cal R}(P)$ on the point set ${\cal R}(S)$.
\end{observation}

Finally, we define the operation of mirroring. Let $P$, $S$, $\cal E$, and $\Gamma$ be as before. Consider a vertical mirroring of $\Gamma$ through a vertical line not separating the points of $S$. This mirrored drawing represents a direction-consistent embedding, denoted by ${\cal M}({\cal E})$, of a new path, denoted by ${\cal M}(P)$, on the mirrored point set, denoted by ${\cal M}(S)$. This new path ${\cal M}(P)$ is formally defined as follows:  ${\cal M}(P)={\cal M}(d_1){\cal M}(d_2)\dots {\cal M}(d_{n-1})$, where ${\cal M}(U)=U$, ${\cal M}(D)=D$, ${\cal M}(R)=L$, and ${\cal M}(L)=R$.
\begin{observation}
\label{obs:mirroring}
Let $\cal E$ be a PDCE of a path $P$ on a point set $S$. Then ${\cal M}({\cal E})$ is a PDCE of ${\cal M}(P)$ on the point set ${\cal M}(S)$.
\end{observation}

\section{Preliminaries}
\label{sec:preliminaries}

In this work we prove that every $n$-vertex three-directional path $P$ admits a PDCE on any set of $n$ points in convex position. As an overview, we sketch the basic idea of the proof.
First, we show that it is possible to construct a PDCE of an $\{U,D,R\}$-path on a one-sided point set, while controlling the position of one of its end-points (Lemma~\ref{lemma:left-sided-UDR} and Lemma~\ref{lemma:right-sided-UDR}).  Then we show that we can embed a two-directional $\{U,R\}$-path on a strip-convex point set $S$ while controlling the positions of both end vertices of the path (Lemma~\ref{lemma:upright}). We use these results to show that an $\{U,D,R\}$-path admits an embedding on any convex point set (Lemma~\ref{lemma:UDR}). For this, we separate a given convex point set into one-sided point sets and a strip-convex point set and go through a case distinction on the labels of the edges which correspond to the separation of the point set. Finally, we show that an embedding of any three-directional path can be reduced to the embedding of an $\{U,D,R\}$-path (Theorem~\ref{theorem:main}). We discuss the direction-consistency of constructed embeddings in detail in the flow of the proofs. However, the planarity of the embedding always follows
from a single simple principle that is described by the following lemma and which is based on Lemma 3 of Binucci~et al.~\cite{BinucciGDEFKL10}.

\begin{lemma}
	\label{lemma:planar}
	 An embedding of an $n$-vertex path on a convex point set is planar if and only if for each $i,~1<i<n$, path $P_{1,i}$ is mapped to a consecutive subset of $S$.
\end{lemma}


\begin{proof}
Let $\mathcal{E}$ be an embedding of $P$ on $S$. Lemma~3 in~\cite{BinucciGDEFKL10} states that if $\mathcal{E}$ is planar then for any~$i,~1<i<n$, both $P_{1,i-2}$ and $P_{i+1,n-1}$ are mapped to consecutive subsets of $S$.

For the reversed direction, assume for the sake of contradiction that $\mathcal{E}$ is not planar. This means that there exists a smallest $j$ such that $(v_j, v_{j+1})$ is crossed by another edge $(v_k, v_{k+1})$, for $k > j$. Vertex $v_1$ lies on $S$ either between $v_k$ and $v_j$ or between $v_j$ and $v_{k+1}$, since $j$ is the smallest index such that $(v_j, v_{j+1})$ is crossed. In both cases, $\mathcal{E}(P_{1,j})$ is not a consecutive subset of $S$, which is a contradiction. 
\end{proof}

We next show that Algorithm~{\sc Backward Embedding} is able to accomplish two tasks: to construct a PDCE of an $\{U,D,R\}$-path on a left-sided point set, and to construct a PDCE of an $\{U,R\}$-path on a strip-convex point set. The algorithm traverses the path backwards and places the vertex $v_i$, $1<i\leq n$, so that, wherever vertex $v_{i-1}$ is placed, edge $(v_{i-1},v_i)$ is guaranteed to be direction-consistent. The algorithm is a generalization of the algorithm constructing a PDCE of an $\{U,D\}$-path~\cite{BinucciGDEFKL10}.

\begin{algorithm}[tpb]
\DontPrintSemicolon 
\KwIn{$\{U,D,L,R\}$-path $P = d_1, \ldots, d_{n-1}$, convex point set $S$ of size $n$}
\KwOut{Function $\mathcal{E} : V(P) \to S$}
\For{$i \gets n-1$ \textbf{downto} $1$} {
  \Switch{$d_i$} {
  		\lCase{$U$:}{$\mathcal{E}(v_{i+1}) \gets t(S)$}
  		\lCase{$D$:}{$\mathcal{E}(v_{i+1}) \gets b(S)$}
  		\lCase{$L$:}{$\mathcal{E}(v_{i+1}) \gets \ell(S)$}
  		\lCase{$R$:}{$\mathcal{E}(v_{i+1}) \gets r(S)$}
  }
  $S \gets S \backslash \{ \mathcal{E}(v_{i+1})\}$\;
}
$\mathcal{E}(v_1) \gets v \in S$ \CommentSty{          //$S$ contains only one element}\;
\Return{$\mathcal{E}$}\;
\caption{{\sc Backward Embedding}}
\label{algo}
\end{algorithm}

\begin{lemma}
\label{lemma:left-sided-UDR} 
Let $S$ be a left-sided point set and let $P = d_1,\ldots,d_{n-1}$ be an $\{U,D,R\}$-path. Algorithm {\sc Backward Embedding} computes a PDCE $\cal E$ of $P$ on $S$ such that $\mathcal{E}(v_n)$ is $t(S)$, $b(S)$, or $r(S)\in \{t(S),b(S)  \}$, dependent on whether $d_{n-1}$ is $U$, $D$, or~$R$, respectively.
\end{lemma}

\begin{proof}
Observe that the algorithm traverses the path backwards and decides the placement of vertex $v_{i+1}$ based on the label of the edge $(v_{i},v_{i+1})$, i.e., $d_i$. If $d_i=U$ (resp. $D,~L,~R$), vertex $v_{i+1}$ is placed on the topmost (resp. bottommost, leftmost, rightmost) of the currently free points. Hence, when vertex $v_{i}$ is placed at the next step on any other free point, edge $(v_{i},v_{i+1})$ is guaranteed to be direction-consistent.

For the planarity, observe that the procedure picking the rightmost, the topmost, and the bottommost points of a left-sided point set, creates a consecutive subset of $S$. Thus, for any $i,~1\leq i \leq n-1$, path $P_{i,n-1}$ (and therefore also $P_{1,i-1}$) is mapped to a consecutive subset of $S$. Hence, by Lemma~\ref{lemma:planar}, the created embedding is also planar. 
\end{proof}

The following lemma can be proven based on Lemma~\ref{lemma:left-sided-UDR} and the operations of rotation of a point set and reverse of a path.
\begin{lemma}
\label{lemma:right-sided-UDR} 
An $\{U,D,R\}$-path $P = d_1,\ldots,d_{n-1}$ admits a PDCE on any right-sided point set $S$ such that $\mathcal{E}(v_1)$ is $b(S),~t(S)$, or $\ell(S)\in \{t(S),b(S) \}$, dependent on whether $d_1$ is $U,~D$, or $R$, respectively.
\end{lemma}

\begin{proof}
Observe that the point set ${\cal R}^2(S)$, i.e., $S$ rotated by $\pi$, is a left-sided point set. Observe also that ${\cal R}^2(P)$ is an $\{U,D,L\}$-path. The reverse of ${\cal R}^2(P)$, i.e., ${\cal R}^2(P)^{\cal I}$, is again an $\{U,D,R\}$-path.
Let $\cal E$ be a PDCE of ${\cal R}^2(P)^{\cal I}$ on ${\cal R}^2(S)$, which exists by Lemma~\ref{lemma:left-sided-UDR}, such that the last vertex of ${\cal R}^2(P)^{\cal I}$ is mapped to $t({\cal R}^2(S))$, $b({\cal R}^2(S))$, or $r({\cal R}^2(S))$ if the last edge of ${\cal R}^2(P)^{\cal I}$ has label $U$, $D$, or $R$, respectively. By Observation~\ref{obs:reverse}, $\cal E ^{\cal I}$ is a PDCE of  ${\cal R}^2(P)$ on ${\cal R}^2(S)$ and finally, by Observation~\ref{obs:rotation}, ${\cal R}^2({\cal E} ^{\cal I})$ is a PDCE of $P$ on $S$. Moreover, observe that the first vertex of $P$ is the last vertex of ${\cal R}^2(P)^{\cal I}$ and that the first edge of $P$ and the last edge of ${\cal R}^2(P)^{\cal I}$ have the same label.
Observe also that the topmost (resp. bottommost, leftmost) point of $S$ is the bottommost (resp. topmost, rightmost) point of ${\cal R}^2(S)$. Hence, we infer that ${\cal R}^2({\cal E} ^{\cal I})(v_1)=b(S)$ if $d_1=U$, ${\cal R}^2({\cal E} ^{\cal I})(v_1)=t(S)$ if $d_1=D$, and ${\cal R}^2({\cal E} ^{\cal I})(v_1)=\ell(S)$ if $d_1=R$. 
\end{proof}

\begin{lemma}
\label{lemma:upright}
Let $S$ be a strip-convex point set and let $P = d_1,\ldots,d_{n-1}$ be an $\{U,R\}$-path. Algorithm~{\sc Backward Embedding} computes a PDCE $\mathcal{E}$ of $P$ on $S$ such that (i) $\mathcal{E}(v_1)$ is $b(S)$ or $l(S)$, and (ii) $\mathcal{E}(v_n)$ is $t(S)$ or $r(S)$, dependent on whether $d_{n-1}$ is $U$ or $R$, respectively.
\end{lemma}
\begin{proof}
Direction consistency of the embedding can be seen similarly to the first part of the proof of Lemma~\ref{lemma:left-sided-UDR}. For the planarity recall that since $S$ is a strip-convex point set, its rightmost and topmost points are either consecutive or coincide. Since $P$ is an $\{U,R\}$-path, Algorithm~{\sc Backward Embedding} picks at every step either the rightmost  or the topmost point of the remaining free points. Thus, the set of used points is a consecutive subset of $S$, and, by Lemma~\ref{lemma:planar}, the embedding is planar.
The position of $v_n$ follows trivially. For the position of $v_1$ we observe the following.  If the algorithm picks $b(S)$ (resp. $\ell(S)$) when searching for the topmost (resp. rightmost) free point then all other points of $S$ have already been used and therefore all the remaining vertices of $P$ except for $v_1$ have already been placed. Hence, $v_1$ is then placed on $b(S)$ (resp. $\ell(S)$). Otherwise, if the algorithm picks
$\ell(S)$ (resp. $b(S)$) when searching for the topmost (resp. rightmost) point of $S$, then the point which is clockwise after $\ell(S)$ (before $b(S)$) has already been used, since it is higher than $\ell(S)$ (resp. to the right of $b(S)$). Therefore, after $\ell(S)$ (resp. $b(S)$) has been used, $b(S)$ (resp. $\ell(S)$) becomes the leftmost (resp. bottommost) free point. Being simultaneously the leftmost and the bottommost free point, $b(S)$ (resp. $\ell(S)$) will be used as the last point by the algorithm.       
\end{proof}

\section{Three-directional paths}
\label{sec:three-dir}
The following lemma is the key ingredient for the proof of a main result of this paper. We postpone its proof until we have seen how the lemma is used.

\begin{lemma}
\label{lemma:UDR}
Let $S$ be a convex point set with the property that $t(S)$ is to the right of $b(S)$. Any $\{U,D,R\}$-path admits a PDCE on $S$.
\end{lemma}

\begin{theorem}
\label{theorem:main}
Any three-directional path admits a PDCE on a convex point set.
\end{theorem}
\begin{proof}
Consider the four different possibilities of a 3-directional path $P$.

\noindent{\bf Case~1:} $P$ is an $\{U,D,R\}$-path.
 Since $S$ is in general position, $t(S)$ is either to the right or to the left of $b(S)$. In the former case a PDCE of $P$ on $S$ exists by Lemma~\ref{lemma:UDR}. For the latter case, observe that in ${\cal M} (S)$, point $t({\cal M} (S))$ is to the right of $b({\cal M} (S))$. Moreover, $P ^{\cal I}$ is an $\{U,D,L\}$-path, and ${\cal M} (P ^{\cal I})$ is again an  $\{U,D,R\}$-path. By Lemma~\ref{lemma:UDR}, there exists a PDCE $\cal E$ of ${\cal M} (P ^{\cal I})$ on ${\cal M} (S)$. By Observation~\ref{obs:mirroring}, ${\cal M}({\cal E})$ is a PDCE of $P ^{\cal I}$ on $S$. Due to Observation~\ref{obs:reverse}, ${\cal M}({\cal E})^{\cal I}$ is a PDCE of $P$ on $S$.

\vspace{0.3ex}\noindent{\bf Case~2:} $P$ is an $\{U,D,L\}$-path.
Observe that $P^{\cal I}$ is an $\{U,D,R\}$-path. Let $\cal E$ be a PDCE of $P^{\cal I}$ on $S$, which exists by Case~1. Then ${\cal E} ^ {\cal I}$ is a PDCE of $P$ on $S$.

\vspace{0.3ex}\noindent{\bf Case~3:} $P$ is an $\{U,L,R\}$-path.
Thus, ${\cal R}(P)$ is an $\{U,D,L\}$-path. Due to Case~2, there exists a PDCE $\cal E$ of ${\cal R}(P)$ on ${\cal R}(S)$. By Observation~\ref{obs:rotation}, ${\cal R}({\cal E})$ is a PDCE of $P$ on $S$.

\vspace{0.3ex}\noindent{\bf Case~4:} $P$ is a $\{D,L,R\}$-path.
Notice that ${\cal R}(P)$ is an $\{U,D,R\}$-path. Thus, for a PDCE $\cal E$ of ${\cal R}(P)$ on ${\cal R}(S)$, which exists due to Case~1, ${\cal R}({\cal E})$ is a PDCE of $P$ on $S$. This concludes the proof of the theorem. 
\end{proof}

\begin{proof}[Proof of Lemma~\ref{lemma:UDR}]
Let $S_\ell$ denote the subset of $S$ containing all points on the left of the line through $b(S)$ and $t(S)$, and let $m=|S_\ell|$.
We distinguish several cases based on the labels $d_m$ and $d_{m+1}$.

\noindent{\bf Case 1:} $d_{m}=D$,  $d_{m+1} \in \{U,R\}$  (see Fig.~\ref{fig:case1} for an illustration).
We embed $P_{1,m}$ on $S_l \cup \{b(S)\}$ using Algorithm~{\sc Backward Embedding}. By Lemma~\ref{lemma:left-sided-UDR}, vertex $v_{m+1}$ is mapped to $b(S)$.
Then, we embed  $P_{m+1,n-1}$ on $S_r \cup \{t(S),b(S)\}$ in the way given by Lemma~\ref{lemma:right-sided-UDR}. Since  $\ell(S_r \cup \{t(S),b(S)\}) = b(S_r \cup \{t(S),b(S)\}) = b(S)$ and $d_{m+1} \in \{U,R\}$, vertex $v_{m+1}$ is mapped to $b(S)$. Thus, the union of these embeddings is a PDCE of $P$ on $S$.

\begin{figure}[tpb]
\centering
\subfigure[]{\label{fig:case1}\includegraphics[scale=0.8]{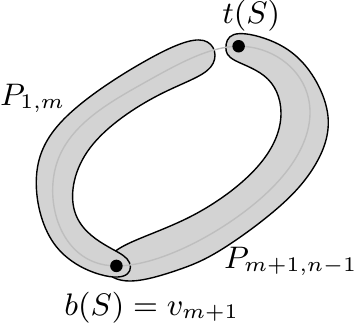}}
\hfill
\subfigure[]{\label{fig:case2}\includegraphics[scale=0.8]{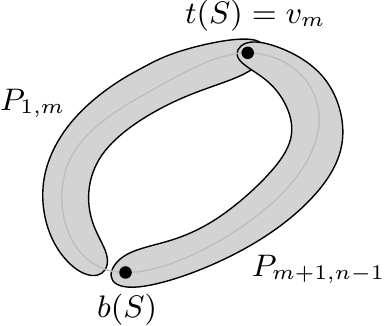}}
\hfill
\subfigure[]{\label{fig:case3}\includegraphics[scale=0.8]{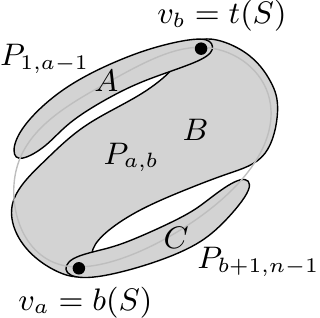}}
\caption{Illustration of the construction in Cases~1-3.}
\label{fig:case1_3}
\end{figure}

\vspace{0.3ex}\noindent{\bf Case 2:} $d_{m}  \in \{U,R\}$, $d_{m+1}=D$ (see Fig.~\ref{fig:case2}).
We embed $P_{1,m}$ on $S_l \cup \{t(S)\}$ using Algorithm~{\sc Backward Embedding}. By Lemma~\ref{lemma:left-sided-UDR}, vertex $v_{m+1}$ is mapped to $t(S)$ since $r(S_l \cup \{t(S)\}) = t(S_l \cup \{t(S)\}) = t(S)$ and $d_m \in \{U,R\}$.
Due to Lemma~\ref{lemma:right-sided-UDR}, we can embed $P_{m+1,n-1}$ on $S_r \cup \{t(S),b(S)\}$ such that vertex $v_{m+1}$ is mapped to $t(S)$,
since $t(S_r \cup \{t(S),b(S)\}) = t(S)$ and $d_{m+1} = D$.
Thus, the union of these embeddings is a PDCE of $P$ on $S$.

\vspace{0.3ex}\noindent{\bf Case 3:} $d_{m}=D$, $d_{m+1}=D$ (see Fig.~\ref{fig:case3}).
Let $P_{a,b}$, $1 \leq a \leq m < m+1 \leq b \leq n-1$, be the maximal subpath of $P$ containing $d_m$, $d_{m+1}$ and only $D$ labels.
Let $A$ be the $a$ highest points of $S_l \cup \{t(S)\}$. Observe that $A$ exists since $a \leq m$.
We embed $P_{1,a-1}$ on $A$ using Algorithm~{\sc Backward Embedding}.
By Lemma~\ref{lemma:left-sided-UDR}, vertex $v_{a}$ is mapped to $t(S)$, since $d_{a-1} \in \{U,R\}$ and $r(A)= t(A) = t(S)$.
Let $C$ be the $n-b$ lowest points of $S_r \cup \{b(S)\}$. Since $|S_r \cup \{b(S)\}| = n-m-1$, and $b\geq m+1$, thus $n-b\leq n-m-1$, and therefore $C$ exists. 
By Lemma~\ref{lemma:right-sided-UDR},  we can embed $P_{b+1,n-1}$ on $C$ such that $v_{b+1}$ is mapped to $b(S)$ since $\ell(C) = b(C) = b(S)$ and $d_{b+1} \in \{U,R\}$.
Let $B$ be $(S \backslash (A \cup C)) \cup \{t(S),b(S)\}$. We embed the $D$-path $P_{a,b}$ on $B$, starting with $v_a$ at $t(S)$ and ending with $v_{b+1}$ at $b(S)$, by sorting the points of $B$ by decreasing $y$-coordinate. Merging the PDCEs for $P_{1,a-1}, P_{a,b}$, and $P_{b+1,n-1}$, we obtain a PDCE of $P$ on $S$.

\vspace{0.3ex}\noindent{\bf Case 4:} $d_m,d_{m+1} \in \{U,R\}$.
Let $P_{i,j}$ where $1 \leq i \leq m < m+1 \leq j \leq n-1$ be the maximal subpath of $P$ containing $d_m, d_{m+1}$ and only $U\!/\!R$-labels. Thus $d_{i-1}\!=\!d_{j+1}\!=\!D$, if they exist.
Let $\alpha$ (resp. $\beta$) denote the number of points of $S$ lying to the left of $b(S)$ (resp. $t(S)$, including $t(S)$).
We consider several cases based on how the indices $i$, $j$ are related to the indices $\alpha$, $\beta$.
The intuition behind this is to distinguish whether or not the points that are vertically between $b(S)$ and $t(S)$ are enough to embed $P_{i,j}$.

\begin{figure}[tpb]
\centering
\subfigure[]{\label{fig:case4ab_path}\includegraphics[scale=0.8]{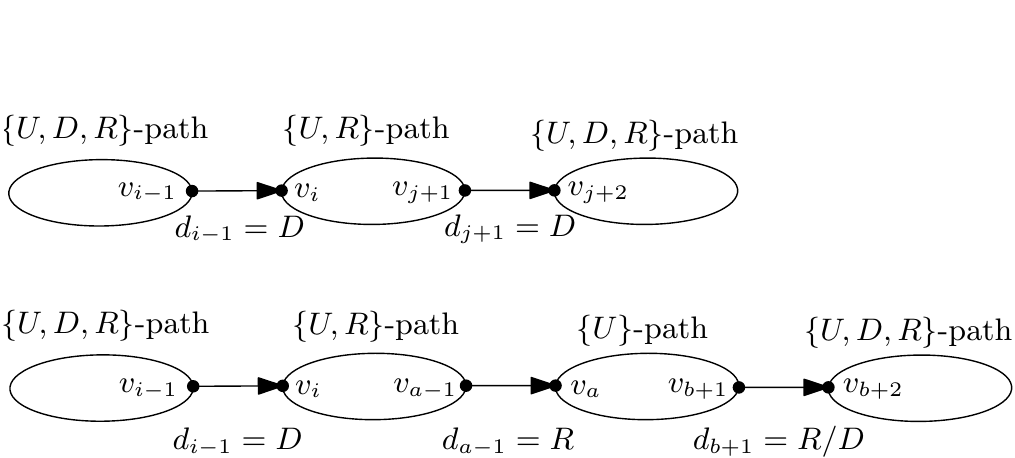}}
\hfill
\subfigure[]{\label{fig:case4a_embedding}\includegraphics[scale=0.8]{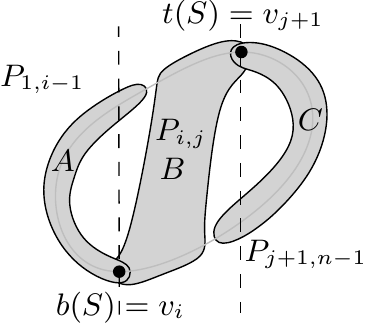}}
\caption{(a) Structure of the path in Cases 4A (above) and 4B
  (below). (b) Construction in Case~4A.}
\label{fig:case4a}
\end{figure}

\vspace{0.3ex}\noindent{\bf Case 4A:} $i > \alpha$ and $j < \beta$, i.e., the points vertically between $b(S)$ and $t(S)$ are enough to embed $P_{i,j}$ (see Fig.~\ref{fig:case4a}). \\
Let $A$ be the $i$ lowest points of $S_l \cup \{b(S)\}$; $A$ exists since $i\leq m$. By Lemma~\ref{lemma:left-sided-UDR}, we can embed $P_{1,i-1}$ on $A$ such that $v_i$ is mapped to $b(S)$. 
Let $C$ be the $n-j$ highest points of $S_r \cup \{t(S)\}$; $C$ exists since $n-j < n-m$. 
By Lemma~\ref{lemma:right-sided-UDR}, we can embed $P_{j+1,n-1}$ on $C$ such that $v_{j+1}$ is mapped to $t(S)$ since $d_{j+1}=D$. Let $B$ be $(S \backslash (A \cup C)) \cup \{b(S),t(S)\}$. Since $i > \alpha$, $\ell(B) = b(B) = b(S)$, and since $j < \beta$, $r(B) = t(B) = t(S)$. Thus, $B$ is a strip-convex point set.
By Lemma~\ref{lemma:upright}, we can embed the $\{U,R\}$-path $P_{i,j}$ on $B$ such that $v_i$ lies on $b(S)$ and $v_{j+1}$ lies on $t(S)$.
By merging the constructed embeddings of $P_{1,i-1},~P_{i,j}$, and $P_{j+1,n-1}$, we obtain a PDCE of $P$ on $S$.\\
Observe that if either $i-1=\alpha$ and $d_{\alpha}=R$ or $j+1=\beta$ and $d_{\beta}=R$ or both, then the embedding can be constructed identically. In case $d_{\alpha}=R$, vertex $v_i$ is mapped to $r(A)=b(S)$. In case $d_{\beta}=R$, vertex $v_{j+1}$ is mapped to $\ell(C)=t(S)$. Thus, these embeddings can be merged with the above embedding of $P_{i,j}$ on $B$.

\begin{figure}[tpb]
\centering
\subfigure[]{\label{fig:case4b_embedding}\includegraphics[scale=0.8]{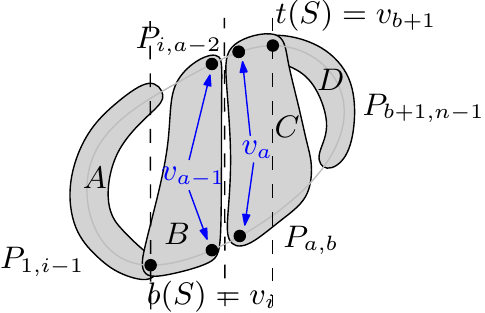}}
\hfill
\subfigure[]{\label{fig:case4c_embedding}\includegraphics[scale=0.8]{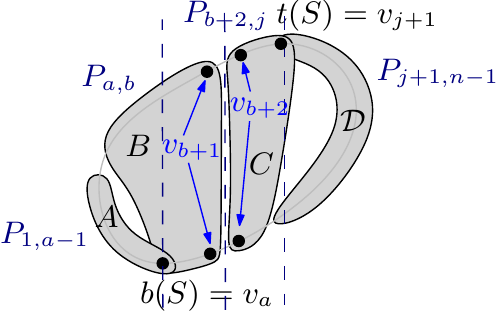}}
\hfill
\subfigure[]{\label{fig:case4d_embedding}\includegraphics[scale=0.8]{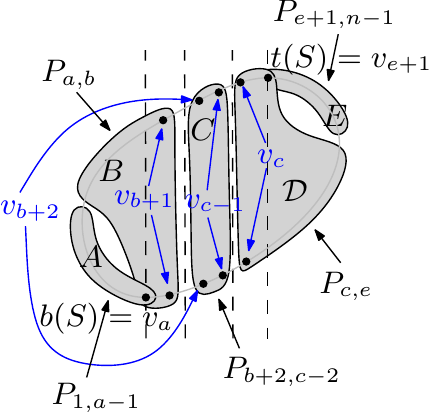}}
\caption{Constructions for (a) Case~4B, (b) Case~4C, and (c) Case~4D when $P_{b+2,c-2}$ is non-empty (if $c=b+2$ the set $C$ is empty; if $a=c$ and $b=e$ the sets $B$ and $\cal D$ are not distinguished).
}
\label{fig:case4bc}
\end{figure}

\vspace{0.3ex}\noindent{\bf Case 4B:} $i > \alpha$ and $j \geq \beta$.
In this case $d_{\beta} \in \{U,R\}$.
If $d_{\beta}=R$ then the embedding is constructed as explained at the end of Case~4A.
In the following we assume $d_{\beta}=U$. \\
Let $P_{a, b}, i \leq a \leq \beta \leq b \leq j$ be the maximal subpath of $P$ containing $d_{\beta}$ and only $U$-edges; see Fig.~\ref{fig:case4ab_path}(below) for the structure of the constructed path.
If $a>i$, $d_{a - 1} = R$. Otherwise, if $a=i$ then $d_{a-1}=D$, i.e., the $\{U,R\}$-path $P_{i,a-2}$ is empty.
Let $A$ be the $i$ lowest points of $S_l \cup \{b(S)\}$ (see Fig.~\ref{fig:case4b_embedding}). Notice that $A$ is a left-sided point set and $b(A) = b(S)$.
We can embed $P_{1,i-1}$ on $A$ by Lemma~\ref{lemma:left-sided-UDR} such that vertex $v_i$ is mapped to $b(S)$.
Let $\cal D$ be the $n - b$ highest points of $S_r \cup \{t(S)\}$. By Lemma~\ref{lemma:right-sided-UDR}, we can embed $P_{b+1,n-1}$ on $\cal D$ such that
vertex $v_{b+1}$ is mapped to $t(S)$.
Let $B$ be the $a - i$ leftmost points of $(S \backslash A) \cup \{b(S)\}$.
If $a=i$ then $B$ is empty. Otherwise, since $i>\alpha$, $\ell(B) = b(B) = b(S)$  and since $a\leq \beta$, the points $t(B)$ and $r(B)$ are consecutive in~$B$.  Thus, $B$ is a strip-convex point set and by Lemma~\ref{lemma:upright} we can embed the $\{U,R\}$-path $P_{i, a - 2}$ on $B$ such that vertex $v_i$ is mapped to $b(S)$ and vertex $v_{a-1}$ is mapped to either $t(B)$ or $r(B)$.
Let $C = S \backslash (A \cup B \cup {\cal D}) \cup \{t(S)\}$.
We embed $P_{a, b}$ on $C$ by sorting the points by increasing $y$-coordinate. Thus, vertex $v_a$ is mapped to $b(C)$ and vertex $v_{b+1}$ is mapped to $t(S)$. 
If $a=i$, vertex $v_i=v_a$ is already mapped to $b(S)$, thus at this step we only embed the vertices of the $\{U\}$-path $P_{a+1,b}$.  \\
Next we merge the constructed PDCEs of $P_{1, i-1}, P_{i, a - 2}, P_{a, b}$, and $P_{b + 1, n-1}$.
If $a=i$, the edge $d_i$ points upward since $v_i$ is mapped to $b(S)$.
Otherwise, since $v_{a - 1}$ is mapped to $t(B)$ or $r(B)$, $v_a$ is mapped to $b(C)$, $B$ and $C$ are separable by a vertical line, and edge $(v_{a-1},v_a)$ points to the right and does not cross the remaining drawing.\\
Recall that this case considers the situation where $i>\alpha$. In case $i\leq \alpha$, we know that $d_{\alpha} \in \{U,R\}$. If it happens that $d_{\alpha}=R$, the construction can be accomplished identically by considering index $\alpha+1$ everywhere in place of $i$. Here, Lemma~\ref{lemma:left-sided-UDR} guarantees a mapping of $P_{1,\alpha}$ with $v_{\alpha+1}$ on $b(S)$ since it is the rightmost point of $A$ and $d_{\alpha}=R$.

\vspace{0.3ex}\noindent{\bf Case 4C:} $i \leq \alpha$ and $j < \beta$. This case is symmetric to  Case~4B. If $d_\alpha=R$ the embedding is constructed as explained at the end of  Case~4A. Otherwise $d_{\alpha}=U$ and we again identify the maximal $\{U\}$-subpath $P_{a,b}$ of $P$ containing $d_\alpha$. The structure of the path in this case is shown in Fig.~\ref{fig:case4cd_path} and the embedding in Fig.~\ref{fig:case4c_embedding}. \\
Also, similar to Case~4B, we can use this construction to embed a path where $j \geq \beta$ and $d_{\beta}=R$. For that, consider the illustration of Fig.~\ref{fig:case4c_embedding}. We set $\cal D$ to contain only points to the right of $t(S)$ and $t(S)$, i.e., $|{\cal D}|=n-\beta+1$. We embed $P_{\beta,n-1}$ on $\cal D$. By  Lemma~\ref{lemma:right-sided-UDR}, we can map $v_{\beta}$ to $t(S)$, since $d_{\beta}=R$ and $t(S)$ is the leftmost point of $\cal D$. The remaining construction is identical.

\begin{figure}[htb]
		\centering
		\includegraphics[scale=0.8]{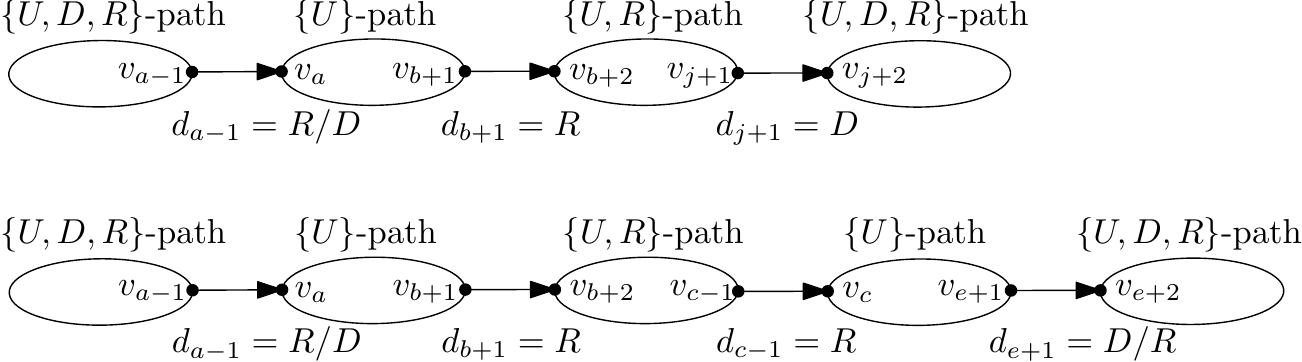}
		\caption{Structure of the path in Cases~4C (above) and 4D (below).}
		\label{fig:case4cd_path}		
\end{figure}

\noindent{\bf Case 4D:} $i \leq \alpha$ and $j \geq \beta$, $d_\alpha=d_\beta=U$.
Let $P_{a,b}$, $a\leq \alpha \leq b$, be the maximal \mbox{$\{U\}$-subpath} of $P$ containing $d_{\alpha}$. Similarly, let $P_{c,e}$, $c\leq \beta \leq e$, be the maximal \mbox{$\{U\}$-subpath} of $P$ containing $d_{\beta}$. If there is no $R$-edge between $d_\alpha$ and $d_\beta$ then $a=c$ and $b=e$. If there is a single $R$-edge between them then $c=b+2$. Otherwise, $P_{b+2,c-2}$ is a $\{U,R\}$-path containing at least one vertex; see Fig.~\ref{fig:case4cd_path} for this case. \\
We embed the $\{U,D,R\}$-path $P_{1,a-1}$ on the $a$ lowest points, denoted by $A$, of $S_\ell \cup \{b(S)\}$.
By Lemma~\ref{lemma:left-sided-UDR}, we can map $v_{a}$ to $b(S)$, since the rightmost point of $A$ is $b(S)$  and $d_{a-1} \in \{D,R\}$.
By Lemma~\ref{lemma:right-sided-UDR}, we can embed $P_{e+1,n-1}$ on the $n-e-1$ highest points, denoted by $E$, of $S_r \cup \{t(S)\}$, such that $v_{e+1}$ is mapped to $t(S)$, since it is the leftmost point of $E$ and $d_{e+1} \in \{D,R\}$. Fig.~\ref{fig:case4d_embedding} shows the case where $P_{b+2,c-2}$ is non-empty. However, it presents the idea of the embedding in the remaining cases as well. \\
If $a=c$ and $b=e$ then $P_{a,e}$ is a $\{U\}$-path. We embed it on $S\setminus (A\cup E) \cup \{b(S),t(S)\}$, by sorting the points by increasing $y$-coordinate. This completes the construction of a PDCE of $P$ on $S$. Otherwise, we let $B$ (resp. $\cal D$) be the $b-a+2$ leftmost (resp. $e-c+2$ rightmost) points of $S\setminus (A \cup E) \cup \{b(S),t(S)\}$. We embed $P_{a,b}$ (resp. $P_{c,e}$) on $B$ (resp. $\cal D$)  by sorting its points by $y$-coordinates. \\
If $c=b+2$, the $\{U\}$-paths $P_{a,b}$ and $P_{c,e}$ are joined by a single $R$-edge.
Since $v_{b+1}$ is to the left of $v_{b+2}=v_c$, the constructed embedding yields a direction-consistent embedding of the edge $(v_{b+1},v_{b+2})$ and this completes the construction of a PDCE of $P$ on $S$. Otherwise, $P_{b+2,c-2}$ is an $\{U,R\}$-path that contains at least one vertex and $d_{b-1}=d_{c-1}=R$. We embed $P_{b+2,c-2}$ on the remaining free points, i.e., on the point set $C=S\setminus (A \cup B \cup {\cal D} \cap E)$.
By construction of $B$ and $\cal D$, the set $C$ is separated from the remaining points by vertical lines. Thus, $\ell(C)$ and $b(C)$ are either consecutive or coincide. Similarly, points $t(C)$ and $r(C)$ are either consecutive or coincide. Thus, $C$ is a strip-convex point set. By Lemma~\ref{lemma:upright}, we can embed $P_{b+2,c-2}$ on $C$ such that $v_{b+2}$ is mapped to one of  $\ell(C)$ or $b(C)$, and $v_{c-1}$ to one of $t(C)$ or $r(C)$. As $v_{b+2}$ is mapped to the highest point of $B$ and $v_{c}$ is mapped to the lowest point of $D$, we infer that the obtained embedding of $P$ on $S$ is planar. Since $d_{b+1}=d_{c-1}=R$ and by the fact that $C$ is separated from $B$ and $\cal D$ by vertical lines, it is also direction-consistent.  This concludes the proof of the lemma. 
\end{proof}

\section{Four-directional paths}
\label{sec:four-dir}

The proof of the following theorem 
is based on the counterexample showing that the path $P=LULRDR$ does not admit a PDCE on the convex point set shown in Figure~\ref{fig:counter}.

\begin{figure}[htb]
		\centering
		\includegraphics[scale=0.8]{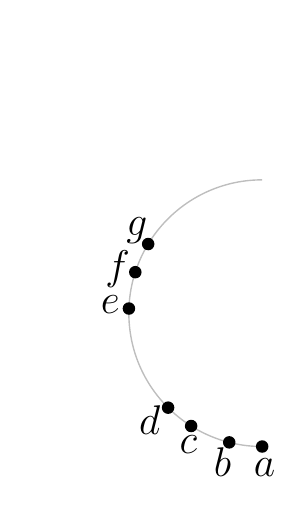}\\[-1ex]
		\caption{Path $P=LULRDR$ does not admit a PDCE on this point set. }
		\label{fig:counter}		
\end{figure}

\begin{theorem}
\label{theorem:4_dir_counter}
	There exists a one-sided point set $S$ and an $\{U,D,L,R\}$-path $P$ such that there is no PDCE of $P$ on $S$.
\end{theorem}

\begin{proof}
Consider the path $P=LULRDR$ and the left-sided point set $S$ of Fig.~\ref{fig:counter}. Lemma~\ref{lemma:planar} states that in order to obtain a planar embedding of $P$ on $S$, a subpath of $P$ must be mapped to consecutive points of $S$. Fig.~\ref{fig:counterproof} illustrates a complete case analysis based on this principle and shows that there is no PDCE of $P$ on $S$.
\begin{figure}[h!]
		\centering
		\includegraphics[scale=1]{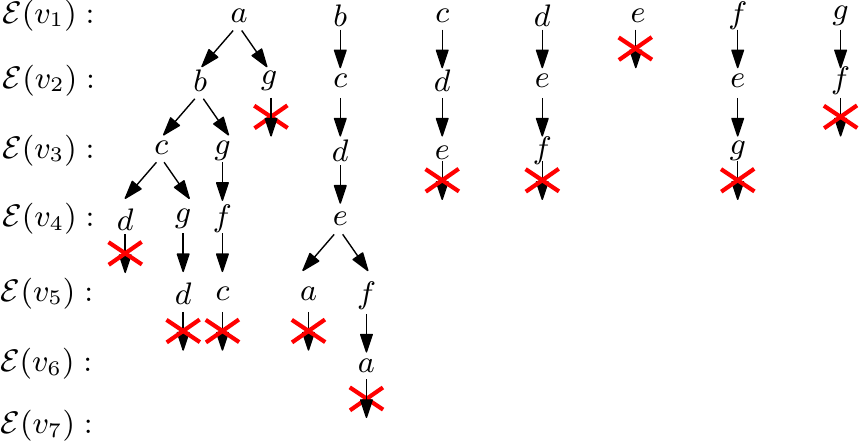}
		\caption{Illustration of the case analysis of the proof of Theorem~\ref{theorem:4_dir_counter}. }
		\label{fig:counterproof}		
\end{figure} 
\end{proof}

A one-sided point set $S$ is a special case of a convex point set, such that  $b(S)$ and $t(S)$ are consecutive. However, as Theorem~\ref{theorem:4_dir_counter} states, such a point set does not always admit a PDCE of every four-directional path.
On the other hand, consider a one-sided convex point set $S$ where one of the following pairs represents a clockwise consecutive subset of $S$:
 $(i)$ $t(S)$ and $\ell(S)$,
 $(ii)$ $r(S)$ and $t(S)$,
 $(iii)$ $b(S)$ and $r(S)$,
 $(iv)$ $\ell(S)$ and $b(S)$.
Such a point set is called \emph{quarter-convex}. It can be easily seen that every quarter-convex point set admits a PDCE of any four-directional path. Actually, in case $(i)$ an edge pointing right always points up and an edge pointing left always points down. Thus, the problem of embedding a $\{U,D,R,L\}$-path
is reduced to embedding a $\{U,D\}$-path, which always admits a PDCE on any convex point set~\cite{BinucciGDEFKL10}. Similar reductions can be made for any other type of a quarter-convex point set. Therefore, we state the following:
\begin{observation}
Any $\{U,D,L,R\}$-path has a PDCE on any quarter-convex point set.
\end{observation}

Based on Lemma~\ref{lemma:planar}, it is easy to derive a dynamic programming algorithm to decide whether a four-directional path admits a PDCE on a convex point set.
This is formalized in the following theorem.
A similar algorithm, described in~\cite{KaufmannMS13}, tests whether an upward planar digraph admits an upward planar embedding on a convex point set.

\begin{theorem}
\label{theorem:4-dir-decide} Let $P$ be an $n$-vertex four-directional path and $S$ be a convex point set. It can be decided in $O(n^2)$ time whether $P$ admits a PDCE on $S$.
\end{theorem}

\begin{proof}
Let $v_1,\dots,v_n$ be the vertices of $P$ and let $t(S)=p_1,\dots,p_n$ be the points of $S$ in counterclockwise order.
Our dynamic programming algorithm stores values $E[i,j]$, which are all possible positions of vertex $v_i$ in a PDCE of $P_{1,i-1}$ (the subpath of $P$ including the first $i$ vertices) on the points $p_j,\dots,p_{j+i-1}$, where $j+i-1$ is taken modulo $n$ if it is greater than $n$. Notice that $E[i,j]$ can contain at most two values, $j$ and $j+i-1$, since these are the only positions for $v_i$ such that the path $P_{1,i-2}$ satisfies the necessary condition of Lemma~\ref{lemma:planar}.

We compute the value $E[i,j]$ as follows. Value $E[i,j]$ contains $j$, if $E[i-1,j+1]$ is non-empty and for at least one of the positions of $v_{i-1}$ given by $E[i-1,j+1]$ the edge $(v_{i-1},v_i)$ is direction-consistent when $v_i$ is placed on $p_j$ (see Fig.~\ref{fig:decision_j} for an illustration of this case). Value $E[i,j]$ contains $j+i-1$, if $E[i-1,j]$ is non-empty and for at least one of the positions of $v_{i-1}$ given by $E[i-1,j]$ the edge $(v_{i-1},v_i)$ is direction-consistent when $v_i$ is placed on $p_{j+i-1}$ (see Fig.~\ref{fig:decision_j_i}).
\begin{figure}
\centering
\subfigure[]{\label{fig:decision_j}\includegraphics[scale=1]{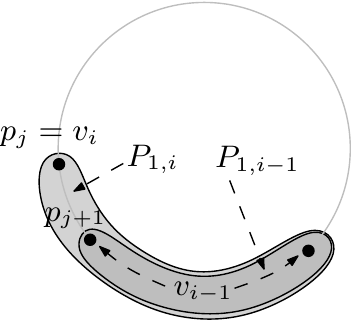}}
\hspace{3cm}
\subfigure[]{\label{fig:decision_j_i}\includegraphics[scale=1]{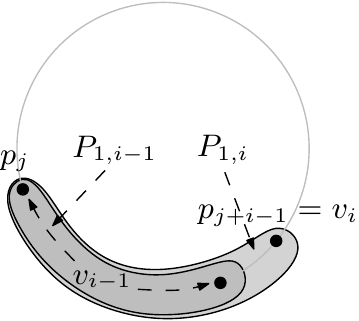}}
\caption{(a) Case when  $E[i,j]$ contains $j$. (b) Case when  $E[i,j]$ contains $j+i-1$.}
\end{figure}

The path $P$ admits a PDCE on $S$ if and only if at least one of the values $E[n,j]$, $1\leq j \leq n$ is non-empty.
Finally, we observe that we need $O(n^2)$ time to compute all the values $E[i,j]$, $1\leq i,j \leq n$. 
\end{proof}

\section{Conclusion}
\label{sec:conclusion}

We investigated the question of finding a planar direction-consistent embedding on a convex point set for any given four-directional path.
We have shown that this is always possible for paths that are restricted to at most three out of the four directions.
To the contrary, we have provided an example showing that for paths using all four directions, this is not always possible.
We also presented an $O(n^2)$ time algorithm to decide embeddability for a given four-directional path and convex point set.

The most challenging open problem is to determine whether any two- or three-directional path always admits a planar direction-consistent embedding on any point set in general position.

\bibliographystyle{abbrv}
\bibliography{bibliography}

\end{document}